\documentclass[a4paper,UKenglish,cleveref, thm-restate]{lipics-v2021}

\pdfoutput=1 
\hideLIPIcs  

\title{Domain Reasoning in TopKAT} 

\author{Cheng Zhang}{Boston University}{czhang03@bu.edu}{https://orcid.org/0000-0002-8197-6181}{National Science Foundation Grant No. 2040249 and No. 2314324.}
\author{Arthur Azevedo de Amorim}{Rochester Institute of Technology}{arthur.aa@gmail.com}{0000-0001-9916-6614}{National Science Foundation Grant No. 2314323.}
\author{Marco Gaboardi}{Boston University}{gaboardi@bu.edu}{0000-0002-5235-7066}{National Science Foundation Grant No. 2040249 and No. 2314324.}
\authorrunning{C. Zhang and A.\,A. de Amorim and M. Gaboardi}

\Copyright{Cheng Zhang and Arthur Azevedo de Amorim and Marco Gaboardi} 



\nolinenumbers 

\EventEditors{Karl Bringmann, Martin Grohe, Gabriele Puppis, and Ola Svensson}
\EventNoEds{4}
\EventLongTitle{51st International Colloquium on Automata, Languages, and Programming (ICALP 2024)}
\EventShortTitle{ICALP 2024}
\EventAcronym{ICALP}
\EventYear{2024}
\EventDate{July 8--12, 2024}
\EventLocation{Tallinn, Estonia}
\EventLogo{}
\SeriesVolume{297}
\ArticleNo{133}

\usepackage[margin=false,inline=true]{fixme}
\FXRegisterAuthor{aaa}{anaaa}{\color{cyan}AAA}

\usepackage{stmaryrd}
\usepackage{mathtools}
\usepackage{mathrsfs}

\usepackage{tikz}
\usetikzlibrary{shapes.geometric, positioning, arrows, matrix}

\usepackage{tikz-cd}
 
\crefname{ineq}{inequality}{inequalities}
\creflabelformat{ineq}{#2{\upshape(#1)}#3} 
\crefname{impl}{implication}{implications}
\Crefname{impl}{Implication}{Implications}
\creflabelformat{impl}{#2{\upshape(#1)}#3} 
\crefname{equiv}{equivalence}{equivalences}
\Crefname{equiv}{Equivalence}{Equivalences}
\creflabelformat{equiv}{#2{\upshape(#1)}#3} 
\crefname{diag}{diagram}{diagrams}
\Crefname{diag}{Diagram}{Diagrams}
\creflabelformat{diag}{#2{\upshape(#1)}#3}

\DeclareMathOperator{\dom}{\mathrm{dom}}
\DeclareMathOperator{\cod}{\mathrm{cod}}
\DeclareMathOperator{\Img}{\mathbf{Im}}

\newcommand{\twhere}{\text{ where }}

\newcommand{\itemTitle}[1]{\textbf{#1}}

\newcommand{\KAT}{\mathsf{KAT}}
\newcommand{\TopKAT}{\mathsf{TopKAT}}

\newcommand{\At}{\mathbf{\mathrm{At}}}
\newcommand{\op}{\mathrm{op}}
\newcommand{\REL}{\mathbf{\mathrm{REL}}}
\newcommand{\TopREL}{\mathbf{\mathrm{TopREL}}}
\newcommand{\TopGREL}{\mathbf{\mathrm{TopGREL}}}

\begin{document}

\maketitle

\begin{abstract}
  TopKAT is the algebraic theory of Kleene algebra with tests (KAT) extended
  with a top element. Compared to KAT, one pleasant feature of TopKAT is that,
  in relational models, the top element allows us to express the domain and
  codomain of a relation.  This enables several applications in program
  logics, such as proving under-approximate specifications or reachability
  properties of imperative programs.  However, while TopKAT inherits many
  pleasant features of KATs, such as having a decidable equational theory, it is
  incomplete with respect to relational models. In other words, there are
  properties that hold true of all relational TopKATs but cannot be proved with
  the axioms of TopKAT.  This issue is potentially worrisome for
  program-logic applications, in which relational models play a key role.

  In this paper, we further investigate the completeness properties of TopKAT
  with respect to relational models. We show that TopKAT is complete with
  respect to (co)domain comparison of KAT terms, but incomplete when comparing
  the (co)domain of arbitrary TopKAT terms. Since the encoding of
  under-approximate specifications in TopKAT hinges on this type of formula, the
  aforementioned incompleteness results have a limited impact when using TopKAT
  to reason about such specifications.
\end{abstract}

\section{Introduction}
  
Kleene algebra with tests (KAT) is an algebraic framework that extends Kleene
algebra with an embedded Boolean algebra to model control structures like
if-statement and while-loops~\cite{Kozen_1997}.  This extension enables us to
reason about several properties of imperative programs.  For example, one of the
key early results in the area was that KAT can encode Hoare logic, in the sense
that any proof in the logic's propositional fragment can be carried out
faithfully using KAT equations~\cite{Manes_Arbib_1986, Kozen_2000}.

Some applications, however, require us to look beyond KAT.  For example, Zhang
et al.~\cite{Zhang_de_Amorim_Gaboardi_2022} recently proved that KAT alone
cannot be used to encode incorrectness logic~\cite{OHearn_2020,
  devries_ReverseHoareLogic_2011}---a close cousin of Hoare logic with
applications in bug finding~\cite{Le_Raad_Villard_Berdine_Dreyer_OHearn_2022,
  Raad_Berdine_Dang_Dreyer_OHearn_Villard_2020}.  A similar result was proved by
Struth~\cite{Struth_2015}, who showed that KAT cannot encode weakest liberal
preconditions. If we view a program as a relation between its input and output
states, both of these limitations arise from KAT's lack of power to encode the
(co)domain of a relation.  Indeed, Möller et
al.~\cite{Möller_O’Hearn_Hoare_2021} proved that incorrectness logic could be
encoded by extending KAT with a codomain operation.  Independently, Zhang et
al. provided a similar encoding~\cite{Zhang_de_Amorim_Gaboardi_2022} by
extending KAT with a top element, which can be used to express inequalities
between codomains.  They dubbed the resulting algebraic structure a TopKAT.

The present paper investigates the expressive power of TopKAT as a tool for
(co)domain reasoning.  As noted by Zhang et al.~\cite{Zhang_de_Amorim_Gaboardi_2022}, 
one limitation of TopKAT is that it is
not expressive enough to derive all valid equations between relations.  More
precisely, Zhang et al.'s encoding of incorrectness logic interprets the top
element of the algebra as the complete relation, which relates all pairs of
program states. Under this interpretation, the inequality \(p  \top  p  \geq  p\) is
valid, but unprovable using the theory of
TopKAT~\cite{Zhang_de_Amorim_Gaboardi_2022}.  This a potential issue when using
TopKAT to reason in incorrectness logic: though Zhang et al.'s encoding covers
all the rules of propositional incorrectness logic, there could be inequalities 
about (co)domain that fall outside this fragment and cannot be established solely
by the theory of TopKAT.

Pous et al.~\cite{Pous_Wagemaker_2022, Pous_Wagemaker_2023} were able to make
some progress on the issue, by showing we can obtain a complete axiomatic system
for relational models TopKATs by adding in the inequality \(p  \top  p  \geq  p\) as an
additional axiom. In this paper, we look at the question from a different angle,
instead of working with a more complex theory,
we show that the original theory of TopKAT is complete with respect to relational models
for (co)domain comparisons, namely the inequalities of the form 
\( \top  t_{1}  \geq   \top  t_{2}\) or \(t_{1}  \top   \geq  t_{2}  \top \) where \(t_{1}, t_{2}\) are KAT terms.  
Since these inequalities suffice to encode incorrectness logic, 
this completeness result lays a solid foundation for encoding program logics in TopKAT.
We have also showed that this completeness result is tight, 
in the sense that it does not extend to the case where \(t_{1}\) and \(t_{2}\) contain the top element, 
by explicitly constructing two TopKAT terms that witness the incompleteness.

The result above is enabled by the homomorphic structure of the
reduction~\cite{Zhang_de_Amorim_Gaboardi_2022,Pous_Rot_Wagemaker_2021} from
TopKAT to KAT.  This discovery also let us shorten the proofs of
previous results~\cite{Zhang_de_Amorim_Gaboardi_2022}, and enables systematic generation of TopKAT complete interpretations from complete interpretations of KAT.  We believe that this new representation of
the reduction technique could also be of independent interest.

\itemTitle{Structure of this paper and contributions:} 
In~\Cref{sec: Preliminaries}, 
we present several previous results on KAT and TopKAT.
Inspired by universal algebra~\cite{burrisCourseUniversalAlgebra1981},
we characterize fundamental concepts, like interpretation and completeness, 
using homomorphisms.
In~\Cref{sec: general completeness},
we uncover additional structure of the reduction technique~\cite{Kozen_Smith_1997,Pous_Rot_Wagemaker_2021}
in the case of TopKAT: the reduction from TopKAT to KAT is a TopKAT homomorphism. 
This discovery not only allows us to simplify several previous 
results~\cite{Zhang_de_Amorim_Gaboardi_2022} by avoiding tedious induction proofs;  
but also enables the techniques used in the later section.
\Cref{sec: domain completeness of TopKAT} presents the completeness results of TopKAT 
with respect (co)domain comparison.
The codomain completeness result is proven by 
an equality that connects codomain operation with the language interpretation,
and the domain completeness is then proven by applying the 
codomain completeness result to the opposite TopKAT.

\section{Preliminaries}\label{sec: Preliminaries}

\subsection{Extensions of Kleene algebra And Their Models}

A \emph{Kleene algebra} is an idempotent semiring with a star operation, written
$p^*$, that satisfies the following \emph{unfolding}, \emph{left induction},
and \emph{right induction} rules:
\[
    p^* = 1 + p p^* = 1 + p^* p, \\
    p r + q  \leq  r  \implies  p^* q  \leq  r, \\
    r p + q  \leq  r  \implies  q p^*  \leq  r;
\]
the ordering here is the conventional ordering in idempotent semirings: \(p  \leq  q  \triangleq  p + q = q.\)
It is known that the right-hand version of unfolding and induction rule 
can be removed while preserving the same equational theory~\cite{Kozen_Silva_2020}.
Yet, we will focus on the standard definition of KA in this paper.
\begin{lemma}\label{the: well known fact about KA}
    Following are well-known facts in Kleene algebra
    \begin{itemize}
        \item All the Kleene algebra operations preserve order.
        \item The following equations are true for the star operation:
              \[ p^*  \cdot  p^* = p^* \\ (p^*)^* = p^*.\]
    \end{itemize}
\end{lemma}

A Kleene algebra with tests (KAT) is a Kleene algebra with an embedded Boolean algebra,
where the conjunction, disjunction, and identities in the Boolean algebra coincide with 
the addition, multiplication, and the identities of Kleene algebra.  
We refer to elements of this embedded Boolean algebra as \emph{tests}.

Given an algebraic theory, we can construct its \emph{free model} 
over a finite set \( \Sigma \), 
called the \emph{alphabet}~\cite{burrisCourseUniversalAlgebra1981}.  
The free model consists of all the terms formed by \( \Sigma \) modulo 
provable equivalences of the algebra. The operations of the free model are obtained 
by lifting the term-level operations to equivalence classes.

The above construction can be extended to the case of KAT and TopKAT, 
suppose that we are given two disjoint finite sets $K$ (the
\emph{action alphabet}) and $B$ (the \emph{test alphabet}).  Elements of $K$ and
$B$ are called \emph{primitive actions and primitive tests}, respectively. 
KAT terms over the alphabet \(K, B\) are defined with the following grammar:
\[t  \triangleq  b  \in  B  \mid  p  \in  K  \mid  1  \mid  0  \mid  t_{1} + t_{2}  \mid  t_{1}  \cdot  t_{2}  \mid  t^*  \mid  \overline{t_b},\]
where \(t_b\) does not contain primitive actions.
The \emph{free KAT} over \(K, B\), written $\KAT_{K,B}$, 
consists of terms over \(K, B\) modulo provable KAT equivalences.  
The tests of the free KAT are Boolean terms, i.e. terms formed by
primitive tests and Boolean operations modulo Boolean axioms.  A similar
construction applies to TopKAT, where an additional symbol \( \top \) was added 
as the largest element in the theory; we denote the free TopKAT over $K,B$ as
\(\TopKAT_{K, B}\).  We sometimes omit the alphabets \(K\) and \(B\) when they
are irrelevant or can be inferred.

In the paper, we frequently consider terms modulo provable equalities, i.e. in the
context of its corresponding free model.  For example, given \(t_{1}, t_{2}  \in  \KAT\),
we will say \(t_{1} = t_{2}\) when they are provably equal using the theory of KAT.
Although the free model seems trivial, it leads to simpler and more modular
proofs of some properties of algebraic theories, as we will see in~\Cref{sec: general completeness}.

Other important models that we will use in this paper are language (Top)KATs and
relational (Top)KATs, which we review here.  An \emph{atom} (short for ``atomic
test'') over a test alphabet \(B = \{b_{1}, b_{2},  \ldots , b_{n}\}\) is a sequence of the form
\[\hat{b_{1}}  \cdot  \hat{b_{2}}  \cdot   \cdots   \cdot \hat{b_{n}} \twhere \hat{b_{i}}  \in  \{b_{i}, \bar{b_{i}}\}.\] 
We denote atoms as \( \alpha ,  \beta ,  \gamma ,  \ldots \) and the set of all atoms as \(\At\).

A \emph{guarded string} (or \emph{guarded word}) over \(K, B\) is an alternation 
between atoms and primitive actions that starts and ends in atoms: 
\[ \alpha _{0}p_{1} \alpha _{1}  \cdots  p_{n}  \alpha _{n} \twhere p_{i}  \in  K,  \alpha _{i}  \in  \At;\] 
where each action is ``guarded'' by an atom.
A guarded string is similar to a program trace, where each program state is denoted by an atom; and primitive actions will cause a transition between program states.
We denote the set of all guarded
strings over alphabet \(K, B\) as \(GS_{K, B}\), and we will omit the alphabet
\(K, B\) when it is irrelevant or can be inferred from context.  The notation
\( \alpha  s\) denotes a guarded string starting with atom \( \alpha \) with the rest of the string
\(s\); similarly, \(s  \alpha \) denotes a guarded string that ends with atom \( \alpha \) with
rest of the string being \(s\).

\begin{definition}[Language/trace KAT~\cite{Kozen_Smith_1997}]
  The \emph{language KAT} (also called ``\emph{trace KAT}'') over an alphabet \(K, B\) is
  denoted as \(\mathcal{G}_{K, B}\), or simply \(\mathcal{G}\) if no confusion can arise.

  The elements are sets of guarded strings (called \emph{guarded languages}), 
  and the tests are sets of atoms.
  The additive identity 0 is the empty set, and the multiplicative identity 1 is
  the set of all the atoms \(\At\).  The addition operator is set union, and the
  multiplication operator is defined as follows:
    \[S_{1}  \diamond  S_{2}  \triangleq  \{s_{1}  \alpha  s_{2}  \mid  s_{1}  \alpha   \in  S_{1},  \alpha  s_{2}  \in  S_{2}\}.\]
    The star operation is defined non-deterministically 
    iterating the multiplication operator:
    \[S^*  \triangleq   \bigcup _{i  \in  \mathbb{N}} S^i \twhere S^0 = \At, S^{k+1} = S  \diamond  S^k.\]
\end{definition}

Another useful type of KAT are relational ones, where each element is a relation
\(R  \subseteq  X  \times  X\) over a fixed set \(X\).  In applications, the set $X$ typically
represents the set of all possible program states, and each relation $R$
represents a program by relating each possible input to the corresponding
output.

\begin{definition}[Relational KAT]
  A relational KAT is a KAT $\mathcal{R}$ consists of relations over a fixed set \(X\) 
  (though $\mathcal{R}$ need not contain every relation over $X$),
  and it is closed under the following operations. 
  The tests are all the relations that are subsets of the identity relation.  
  The additive identity 0 is the empty set, and
  the multiplicative identity is the identity relation:
  \[1  \triangleq  \{(x, x)  \mid  x  \in  X\}.\] The addition operator is set union, and the
  multiplication operation is relational composition:
  \[R_{1} ; R_{2} = \{(x, z)  \mid   \exists  y  \in  X, (x, y)  \in  R_{1}, (y, z)  \in  R_{2}\}.\] 
  Finally, the star operation is defined as:
  \[R^*  \triangleq   \bigcup _{i  \in  \mathbb{N}} R^i \twhere R^0 = 1, R^{k+1} = R ; R^k.\] We denote the
  class of all relational KATs as \(\REL\).
\end{definition}

\emph{TopKAT} extends the theory of KAT with the largest element \( \top \), i.e.
\( \top   \geq  p\) for all elements \(p\).  The \emph{language TopKAT} over an alphabet
\(K, B\) has the same carrier and operations as \(\mathcal{G}_{K_ \top , B}\), where \(K_ \top \) is
the set \(K\) joined with a new primitive action \( \top \); and the largest element
is the full language \(GS_{K_ \top , B}\).

The \emph{relational TopKAT} is a relational KAT that contains the complete relation:
\[ \top   \triangleq  \{(x, y)  \mid  x, y  \in  X\};\] we denote the set of all relational TopKATs as
\(\TopREL\).  It is known that there are equations that are valid in relational
TopKAT, but are not derivable by the axioms of
TopKAT~\cite{Zhang_de_Amorim_Gaboardi_2022}; however, by adding the axiom
\(p  \top  p  \geq  p\), the theory becomes complete over relational
TopKATs~\cite{Pous_Wagemaker_2022,Pous_Wagemaker_2023}.  
In this paper, instead of working with a more complex theory, 
we will show that TopKAT without any additional axiom already suffices 
for the purpose of encoding domain comparisons. 
Indeed, TopKAT is complete with respect to domain comparison inequalities,
which can be used to encode both incorrectness logic and Hoare logic.

In this paper, 
we will use \(\dom\) and \(\cod\) to denote the 
conventional (co)domain operators on relations, namely, for any relation \(R\):
\[
    \dom(R)  \triangleq  \{x  \mid   \exists  y, (x, y)  \in  R\} \\
    \cod(R)  \triangleq  \{y  \mid   \exists  x, (x, y  \in  R)\}.
\]
To demonstrate how TopKAT models (co)domain comparisons,
we take any relational TopKAT \(\mathcal{R}\) and two relations \(R_{1}, R_{2}  \in  \mathcal{R}\),
and we denote the complete relation as \( \top \):
\begin{lemma}[TopKAT encodes (co)domain comparison]\label{the: TopKAT encodes domain PRIMITIVE}
    \[
        R_{1}  \top   \supseteq  R_{2}  \top   \iff  \dom(R_{1})  \supseteq  \dom(R_{2}) \\
         \top  R_{1}  \supseteq   \top  R_{2}  \iff  \cod(R_{1})  \supseteq  \cod(R_{2})
    \]
\end{lemma}
If we regard \(R_{1}\) and \(R_{2}\) as the input output relation of two programs,
which is typically encoded by KAT terms,
we can see that \(R_{1}  \top   \supseteq  R_{2}  \top \) reflects that  
the domain of \(R_{1}\) is larger than the domain of \(R_{2}\);
and similarly for the inequality \( \top  R_{1}  \supseteq   \top  R_{2}\).
Thus, given two KAT terms \(t_{1}, t_{2}  \in  \KAT_{K, B}\), we call inequalities like
\(t_{1}  \top   \geq  t_{2}  \top \) \emph{domain comparison inequalities},
and \( \top  t_{1}  \geq   \top  t_{2}\) \emph{codomain comparison inequalities}.
Notice that the term \( \top  t_{1}\) is a shorthand for \( \top   \cdot  i(t_{1})\),
where \(i\) is the inclusion function \(\KAT_{K, B}  \hookrightarrow  \TopKAT_{K, B}\).
In the rest of the paper, we will sometimes leave this inclusion function implicit.
These two forms of inequalities will be the focus of 
our completeness results in~\Cref{sec: domain completeness of TopKAT}.

We also know another class of TopKATs named \emph{general relational TopKATs},
which is denoted as \(\TopGREL\).
The top element of general relational TopKAT is not necessarily the complete relation,
but the largest relation in the model.
All equations in the general relational TopKAT can be derived using the theory of TopKAT.

However, the completeness of \(\TopGREL\) came at the cost of expressive power:
every predicate that is expressible using general relational TopKAT 
is already expressible using relational KAT~\cite{Zhang_de_Amorim_Gaboardi_2022},
so the extension with top, in the case of general relational TopKAT, 
does not grant any extra expressive power.
In~\Cref{the: TopGREL expressive power}, 
we show that this result is a simple corollary of our new reduction result.

We are also interested in maps between models:
A \emph{KAT homomorphism} \(f\) is a map between two KATs \(\mathcal{K}\) and \(\mathcal{K}'\)
s.t. it preserves the sorts and operations:
given a test \(b\) in \(\mathcal{K}\) then \(f(b)\) is a test in \(\mathcal{K}'\);
and all the KAT operations (complement, identities, addition, multiplication, and star) are preserved:
\begin{align*}
    f & : \mathcal{K}  \to  \mathcal{K}'\\
    f(\bar{b}) & = \overline{f(b)} \\  
    f(1) & = 1 \\  
    f(0) & = 0 \\
    f(p + q) & = f(p) + f(q) \\  
    f(p  \cdot  q) & = f(p)  \cdot  f(q) \\  
    f(p^*) & = f(p)^*.
\end{align*}
Similarly, a \emph{TopKAT homomorphism} is a KAT homomorphism that preserves the
largest element.

\subsection{Interpretation, Completeness, and Injectivity}\label{sec: completeness background}

Consider a KAT equation such as \(p  \cdot  b  \cdot  \bar{b} = 0\). To determine its
validity in a particular KAT \(\mathcal{K}\), we need to assign meaning to it by
interpreting each primitive as an element in \(\mathcal{K}\); that is, by defining a map
\(\hat{I}\) of type \(K + B  \to  \mathcal{K}\).  Such a map \(\hat{I}: K + B  \to  \mathcal{K}\) induces a
unique KAT homomorphism \(I : \KAT_{K,B}  \to  \mathcal{K}\) inductively defined on the term 
as follows:
\begin{equation}
    \begin{aligned}
        I(p)       &  \triangleq  \hat{I}(p)    & \twhere p  \in  K + B \\
        I(\overline{t_b}) &  \triangleq  \overline{I(t_b)} 
            & \text{\(t_b\) does not contain primitive actions} \\
        I(t_{1} + t_{2}) &  \triangleq  I(t_{1}) + I(t_{2})                     \\
        I(t_{1}  \cdot  t_{2}) &  \triangleq  I(t_{1})  \cdot  I(t_{2})                     \\
        I(t^*)     &  \triangleq  I(t)^*
    \end{aligned}
\end{equation}
In fact, every KAT homomorphism from a free model arises this way: there is a
bijection between functions of type \(K + B  \to  \mathcal{K}\) and KAT homomorphisms of type
\(\KAT_{K, B}  \to  \mathcal{K}\), for any KAT \(\mathcal{K}\).  
Because the homomorphism \(I\) and the function \(\hat{I}\) are equivalent, 
we will refer to them interchangeably as \emph{KAT interpretations} 
and denote both of them as \(I\).

The above result enables us to define a homomorphism from the free KAT just by
defining its action on the primitives; saving us time to check the equations
that a homomorphism must satisfy.  It also allows us to prove that two
interpretations are equal by arguing that they map the primitives to
equal values.

Given a KAT \(\mathcal{K}\), and two terms \(t_{1}, t_{2}  \in  \KAT_{K, B}\) we say that \(\mathcal{K}  \models  t_{1} = t_{2}\) if
\[ \forall  I : \KAT_{K, B}  \to  \mathcal{K}, I(t_{1}) = I(t_{2}).\] In particular, 
for two terms in the free model \(t_{1}, t_{2}  \in  \KAT_{K, B}\),
\(\KAT_{K, B}  \models  t_{1} = t_{2}\) is equivalent to \(t_{1} = t_{2}\).  
For a collection of models \(\mathsf{K}\), 
we say that \(\mathsf{K}  \models  t_{1} = t_{2}\) if for all \(\mathcal{K}  \in  \mathsf{K}\),
\(\mathcal{K}  \models  t_{1} = t_{2}\).  For example, \(\REL  \models  t_{1} = t_{2}\) means that \(t_{1} = t_{2}\) is
valid in all relational KATs.  All the above notations and terminologies can be
similarly extended to TopKAT.

Theories like KAT and TopKAT are designed to model practical
programs, so it is important to know if they can model all the desirable
equations between programs. If the theory of KAT can derive all the equalities
for a particular interpretation \(I\), namely:
\[\KAT_{K, B}  \models  t_{1} = t_{2}  \iff  I(t_{1}) = I(t_{2}),\]
we say that the theory of KAT is \emph{complete} with respect to \(I\).
Recall that \(\KAT_{K, B}  \models  t_{1} = t_{2}\) is equivalent to \(t_{1} = t_{2}\);  
thus, by definition, an interpretation \(I\) is complete if and only if it is injective.
One of such interpretation is the guarded string interpretation
\(G: \KAT_{K, B}  \to  \mathcal{G}_{K, B}\)~\cite{Kozen_Smith_1997},
defined by lifting the following action on the primitives:
\[
    G(b) = \{ \alpha   \mid  \text{\(b\) appears positively in \( \alpha \)}\}, \\
    G(p) = \{ \alpha  p  \beta   \mid   \alpha ,  \beta   \in  \At\}.
\]



In several previous works, the term ``free model'' refers to the range (set of
reachable elements) of a complete interpretation.  Since a complete
interpretation is an injective homomorphism, 
such interpretation induces an isomorphism on its range, 
thus our definition of free model is equivalent to these definitions.

Many previous proofs can also be explained by seeing complete interpretations as
injective homomorphisms: the proof for completeness of relational KATs
constructs an injective homomorphism $h$ from a language KAT into a relational
KAT~\cite{Kozen_Smith_1997}.  Since both \(G\) and \(h\) are injective
homomorphisms, \(h  \circ  G\) is also an injective homomorphism, hence a complete
interpretation.  Since \(h  \circ  G\) is a relational interpretation:
\[\KAT_{K, B}  \models  t_{1} = t_{2}  \implies  \REL  \models  t_{1} = t_{2}  \implies  h  \circ  G(t_{1}) = h  \circ  G(t_{2});\]
then the completeness of \(h  \circ  G\) implies
\((h  \circ  G)(t_{1}) = (h  \circ  G)(t_{2})  \iff  \KAT_{K, B}  \models  t_{1} = t_{2}\). Hence,
\[\KAT_{K, B}  \models  t_{1} = t_{2}  \iff  \REL  \models  t_{1} = t_{2},\]
i.e. the theory of KAT is complete with respect to relational KAT.



Besides using composition of injective homomorphisms, another technique commonly
used to prove injectivity is to construct a left inverse: 
if a (Top)KAT homomorphism \(f: \mathcal{K}  \to  \mathcal{K}'\) has a left inverse homomorphism \(g: \mathcal{K}'  \to  \mathcal{K}\) 
i.e. \(g  \circ  f = id_{\mathcal{K}}\), then \(f\) is injective.  
Notice that \(g\) does not need to be a homomorphism for \(f\) to be injective,
however, in the case where \(f\) is an interpretation, 
\(g\) being a homomorphism makes the equality \(g  \circ  f = id_{\mathcal{K}}\) easier to check.
Because both \(g  \circ  f\) and \(id_{\mathcal{K}}\) are all interpretations,
they are equal if and only if they have the same action on all the primitives.

Finally, we provide a shorthand for domain reasoning. 
For two terms \(t_{1}, t_{2}  \in  \KAT\), we write
\[\REL  \models  \dom(t_{1})  \geq  \dom(t_{2}),\] when 
\(\dom(I(t_{1}))  \supseteq  \dom(I(t_{1}))\) for all relational KAT interpretations \(I\);  
and similarly for relational TopKAT and general relational TopKAT.  
Then \Cref{the: TopKAT encodes domain PRIMITIVE} implies the following:
\begin{lemma}\label{the: top element represent domain}
    For two KAT terms \(t_{1}, t_{2}  \in  \KAT_{K, B}\):
    \begin{align*}
        \TopREL  \models  t_{1}  \top   \geq  t_{2}  \top  &  \iff  \REL  \models  \dom(t_{1})  \geq  \dom(t_{2}) \\
        \TopREL  \models   \top  t_{1}  \geq   \top  t_{2} &  \iff  \REL  \models  \cod(t_{1})  \geq  \cod(t_{2})
    \end{align*}
\end{lemma}

\section{Reduction, A New Perspective}
\label{sec: general completeness}

Our goal in this section is to construct a complete interpretation for TopKAT,
by reducing its theory to that of plain KAT.  In other words, any equation
between two TopKAT terms is logically equivalent to another equation between a
pair of corresponding KAT terms.  While this result is not
new~\cite{Zhang_de_Amorim_Gaboardi_2022, Zhang_de_Amorim_Gaboardi_2022_POPL,
  Pous_Wagemaker_2022}, we present a more streamlined proof that hinges on the
universal properties of free KATs and TopKATs, without relying explicitly on
language models.  Similar to previous works, we obtain the
decidability of the equational theory of TopKAT as a corollary of reduction.
However, because of the new notion of reduction,
our decidability result no longer depends on the completeness of the language TopKAT.  
Moreover, our technique helps us to construct complete models and interpretations 
simply by computation, as well as simplifying proofs of other results about TopKAT.


\subsection{Reduction on free models}\label{sec: reduction on free models}

We first note that any free KAT over an alphabet \(K, B\) is also a TopKAT,
where the largest element is \(( \sum  K)^*\). This fact can be seen by
straightforward induction.

\begin{lemma}\label{the: every free KAT is a TopKAT}
    Every free KAT over alphabet \(K, B\) forms a TopKAT.
\end{lemma}

\begin{proof}
    Since \(\KAT_{K, B}\) is a KAT, we only need to show 
    the term \(( \sum  K)^*\) is the largest element of \(\KAT_{K, B}\),
    i.e. \[( \sum  K)^*  \geq  t,  \forall  t  \in  \KAT_{K, B}.\] 
    The above fact can be shown by induction on \(t\);
    some algebraic manipulations below use facts in~\Cref{the: well known fact about KA}:
    \begin{itemize}
        \item \(( \sum  K)^*  \geq  1\) (by unfolding rule),
              thus \(( \sum  K)^*\) is larger than \(0, 1\) and every Boolean term.
        \item \(( \sum  K)^*\) is larger than \( \sum  K\),
              which is larger than every primitive action.
        \item Given two terms \(t_{1}\) and \(t_{2}\),
              assume \(( \sum  K)^*\) is larger than both.
              Because \(( \sum  K)^* = ( \sum  K)^* + ( \sum  K)^*\)
              and addition preserves order,
              \[( \sum  K)^* = ( \sum  K)^* + ( \sum  K)^*  \geq  t_{1} + t_{2}\] 
        \item Given two terms \(t_{1}\) and \(t_{2}\),
              assume \(( \sum  K)^*\) is larger than both.
              Because \(( \sum  K)^* = ( \sum  K)^*  \cdot  ( \sum  K)^*\)
              and multiplication preserves order, 
              \[( \sum  K)^* = ( \sum  K)^*  \cdot  ( \sum  K)^*  \geq  t_{1}  \cdot  t_{2}.\]
        \item Given a term \(t\),
              if \(( \sum  K)^*  \geq  t\), then \(( \sum  K)^*  \geq  t^*\).
              Since \(( \sum  K)^* = (( \sum  K)^*)^*\) and star preserves order:
              \[( \sum  K)^* = (( \sum  K)^*)^*  \geq  t^*. \qedhere\]
    \end{itemize}
\end{proof}

Since every free KAT is a TopKAT, every KAT interpretation
\(I : \KAT  \to  \mathcal{K}\) induces a sub-KAT $\Img(I)  \subseteq  \mathcal{K}$,
and this sub-KAT happens to be a \emph{TopKAT}. Specifically, the image of $( \sum  K)^*$
in $\mathcal{K}$ is the largest element of $\Img(I)$, and the restricted
$I : \KAT  \to  \Img(I)$ is a TopKAT homomorphism.

This gives us a powerful tool to construct complete TopKAT interpretations.
Since we already know that the KAT interpretations \(G: \KAT  \to  \mathcal{G}\) and
\(h  \circ  G: \KAT  \to  \Img(h)\) are injective TopKAT homomorphisms, we can
construct complete TopKAT interpretations by \emph{composition}, 
if we can construct an injective TopKAT interpretation \(r\) of type
\(\TopKAT_{K, B}  \to  \KAT_{K_ \top , B}\):
\[\TopKAT_{K, B} \xrightarrow{r} \KAT_{K_ \top , B} \xrightarrow{G} \mathcal{G}_{K_ \top , B},\\  
  \TopKAT_{K, B} \xrightarrow{r} \KAT_{K_ \top , B} \xrightarrow{G} \mathcal{G}_{K_ \top , B}
  \xrightarrow{h} \Img(h).\] 

In fact, such an injective homomorphism can be obtained by lifting 
the embedding map \(K + B  \hookrightarrow  \KAT_{K_ \top , B}\):
\begin{align*}
    r   & : K + B  \to  \KAT_{K_ \top , B}             \\
    r(p) &  \triangleq  p.                   
\end{align*}
This homomorphism coincides with the \emph{reduction maps} of the same name in
previous works~\cite{Zhang_de_Amorim_Gaboardi_2022, Pous_Wagemaker_2023}.  More
concretely, we can picture $r$ as simply replacing the symbol \( \top \) in a TopKAT
term with \(( \sum  K_ \top )^*\), the largest element in \(\KAT_{K_ \top , B}\).

We will show that \(r\) is injective by constructing a left inverse for it.  
In fact, the left inverse \([-]_ \top \) simply interprets the \( \top \) primitive in \(\KAT_{K_ \top , B}\)
as the largest element.
\begin{lemma}\label{the: equivalence class is the inverse of reduction}
  The map \([-]_ \top : \KAT_{K_ \top , B}  \to  \TopKAT_{K, B}\), where each term is
  mapped to its corresponding equivalence class, 
  is defined by lifting the following action on the primitives:
  \begin{align*}
    [p]_ \top  &  \triangleq  p & \text{if } p  \in  K + B \\  
    [ \top ]_ \top  &  \triangleq   \top .
  \end{align*}
  The map \([-]_ \top \) is a TopKAT homomorphism.
\end{lemma}


\begin{proof}
  Because this map defined by lifting on the primitives,
  it is automatically a KAT homomorphism.
  All we need to show is that \([-]_ \top \) preserves the top element, that is
  \([( \sum  K_ \top )^*]_ \top  = ( \sum  K_ \top )^*\) is the largest element in \(\TopKAT_{K, B}\).

  By construction of \(\TopKAT_{K, B}\), \( \top \) is the largest element in \(\TopKAT_{K, B}\). 
  Thus, to prove that \(( \sum  K_ \top )^*\) is also the largest element in \(\TopKAT_{K, B}\),
  it suffices to prove \(( \sum  K_ \top )^*  \geq   \top \): \[( \sum  K_ \top )^*  \geq   \sum  K_ \top  =  \top  +  \sum  K  \geq   \top . \qedhere\]
\end{proof}

\begin{theorem}[Reduction]
    \([-]_ \top \) is the right inverse of \(r\): \([-]_ \top   \circ  r  = id_{\TopKAT_{K, B}}\).
    More explicitly for all \(t  \in  \TopKAT_{K, B}\): \[\TopKAT_{K, B}  \models  [r(t)]_ \top  = t.\]
\end{theorem}

\begin{proof}
    Since \([-]_ \top   \circ  r : \TopKAT_{K, B}  \to  \TopKAT_{K, B}\) is a TopKAT interpretation,
    the action on the primitives uniquely determines the interpretation:
    because both \(r\) and \([-]_ \top \) are identity on the primitives,
    therefore \([-]_ \top   \circ  r\) is the identity interpretation on \(\TopKAT_{K, B}\).
\end{proof}

The above theorem matches one of the soundness condition of reductions in 
previous works~\cite{Zhang_de_Amorim_Gaboardi_2022,Kozen_Smith_1997,Pous_Rot_Wagemaker_2021},
which was typically proven by a monolithic induction on the structure of terms.
Our approach, on the other hand, relies on establishing fine-grained 
algebraic properties, like~\cref{the: every free KAT is a TopKAT,the: equivalence class is the inverse of reduction};
then the theorem follows simply by computing the action of \([-]_ \top   \circ  r\) on primitives.

Since \(r\) has a right inverse, it is indeed the injective interpretation we desired, 
and it is also a complete interpretation:
\[\TopKAT_{K, B}  \models  t_{1} = t_{2}  \iff  r(t_{1}) = r(t_{2}),\]
With the completeness of \(r\), we can already show the complexity of TopKAT.
The complexity results echos previous proofs~\cite{Zhang_de_Amorim_Gaboardi_2022,Pous_Wagemaker_2023},
but we are able to obtain this result without completeness of TopKAT language interpretation,
which is essential in previous proofs. 

\begin{corollary}[Complexity]\label{the: PSPACE-completeness of TopKAT}
  Given two terms \(t_{1}, t_{2}  \in  \TopKAT_{K, B}\), deciding whether these two terms
  are equal is PSPACE-complete.
\end{corollary}

\begin{proof}
    Deciding KAT equality is a sub-problem of deciding TopKAT equality,
    and KAT equality is PSPACE-hard \cite{Cohen_Kozen_Smith_1999};
    therefore TopKAT equality is PSPACE-hard.

    To decide the equality of \(t_{1}, t_{2}\),
    we first remove all the redundant primitives that do not appear in \(t_{1}, t_{2}\)
    from the alphabet \(K, B\). Then we compute \(r(t_{1})\) and \(r(t_{2})\),
    each taking polynomial space (of \(|t_{1}| + |t_{2}|\)) to store;
    and we use the standard algorithm \cite{Cohen_Kozen_Smith_1999}
    to decide whether \(r(t_{1}) = r(t_{2})\) in \(\KAT_{K_ \top , B}\),
    this will also take polynomial space.
    Hence, the decision procedure for TopKAT equality in PSPACE.

    Thus deciding TopKAT equality is PSPACE-complete.
\end{proof}

\subsection{Computing the complete interpretations}\label{sec: complete model for free}

Designing complete interpretations and models was not always easy.
In fact, in previous works \cite{Zhang_de_Amorim_Gaboardi_2022_POPL},
the authors made a mistake in the definition of language TopKAT,
which was fixed later \cite{Zhang_de_Amorim_Gaboardi_2022} 
by suggestion of Pous et al. \cite{Pous_Wagemaker_2022}.
However, with the results in~\Cref{sec: reduction on free models},
we can construct the complete interpretation just by composition,
and compute the complete model by computing the range of the complete interpretation.

We already know that there are two complete interpretations of TopKAT defined as follows:
\[\TopKAT_{K, B} \xrightarrow{r} \KAT_{K_ \top , B} \xrightarrow{G} \mathcal{G}_{K_ \top , B},\\  
\TopKAT_{K, B} \xrightarrow{r} \KAT_{K_ \top , B} \xrightarrow{G} \mathcal{G}_{K_ \top , B} \xrightarrow{h} \Img(h),\]
with a complete language model \(\mathcal{G}_{K_ \top , B}\), 
and a complete model consisting of relations \(\Img(h)\).

The operations in these models can be recovered by computing these maps.
For example, the multiplication operation in the language TopKAT can be computed as follows:
\[G  \circ  r(t_{1}  \cdot  t_{2}) = G(r(t_{1})  \cdot  r(t_{2})) = G(r(t_{1}))  \diamond  G(r(t_{2})).\]
Since \(r\) does not change the multiplication operation,
the multiplication in the language TopKAT is the same as in language KAT.
In fact, as \(r\) does not change any operation in KAT,
most operations in language TopKAT are the same as language KAT.
Thus, we only need to compute the top element in language TopKAT.

The top element in language TopKAT can be computed in the same fashion:
\[G  \circ  r( \top ) = G(( \sum  K_ \top )^*) = GS_{K_ \top , B},\]
i.e. the top element is just the complete language.

\begin{corollary}\label{the: language TopKAT for free}
    The language TopKAT inherits all the operations in language KAT,
    except the top element, which is defined as the full language.
    And such models are complete with \(G  \circ  r\) as a complete interpretation.
\end{corollary}

In the same way, we know that complete models consisting of relations (a.k.a. general relational TopKAT) 
will have the same operations as relational KATs.
However, in this case the characterization of the computed top: \(h  \circ  G  \circ  r( \top )\)
is not as simple as the full language,
but we know it is the largest relation in the range of \(h  \circ  G  \circ  r\):

\begin{corollary}\label{the: general relational TopKAT for free}
    The general relational TopKAT inherits all the operations in relational KAT,
    except the top element is the largest relation.
    And such models are complete with \(h  \circ  G  \circ  r\) as a complete interpretation.
\end{corollary}

Finally, to investigate whether we can use general relational TopKAT
to encode incorrectness logic,
we will provide a short proof that general relational TopKATs
are as expressive as relational KATs~\cite{Zhang_de_Amorim_Gaboardi_2022};
that is, every property on relations that can be encoded using general relational TopKAT,
is already encodable in the relational KAT.
Hence, adding a top element does not give extra expressive power in general relational TopKAT.


The original proof~\cite[Lemma 2]{Zhang_de_Amorim_Gaboardi_2022} 
encodes every TopKAT term using a KAT term,
and then uses two pages to prove the soundness of this encoding.
Here we show the aforementioned encoding is simply the reduction \(r\).

\begin{definition}
    Given two terms \(t_{1}, t_{2}  \in  \TopKAT\), and n primitives \(p_{1}, p_{2},  \ldots  , p_{n}  \in  K + B\),
    we say that an n-ary predicate \(P\) is \emph{expressible} by 
    equation \(t_{1} = t_{2}\) for a class of TopKATs \(\mathsf{K}\) 
    when for all interpretations \(I\) into TopKATs in \(\mathsf{K}\),
    the following equivalence holds:
    \[I(t_{1}) = I(t_{2})  \iff  P(I(p_{1}), I(p_{2}),  \ldots , I(p_{n})).\]
\end{definition}

\begin{theorem}[Expressiveness of general relational TopKAT]\label{the: TopGREL expressive power}
    Given an alphabet \(K, B\), an n-ary predicate \(P\) on relations,
    the predicate \(P\) over primitives \(p_{1}, p_{2},  \ldots  , p_{n}  \in  K\) is expressible in
    general relational TopKAT if and only if it is expressible in relational KAT.
\end{theorem}

\begin{proof}
    A predicate expressible in relational KAT is also expressible 
    in general relational TopKAT using the same pair of terms,
    we only need to show the converse.
    Assume a predicate \(P\) is expressible in general relational TopKAT,
    then there exists two TopKAT terms \(t_{1}, t_{2}  \in  \TopKAT_{K, B}\) s.t. 
    for all general relational TopKAT interpretations \(I_ \top \):
    \[I_ \top (t_{1}) = I_ \top (t_{2})  \iff  P(I_ \top (p_{1}), I_ \top (p_{2}),  \ldots  , I_ \top (p_{n}));\]

    We take an arbitrary relational KAT interpretation \(I\) from \(\KAT_{K_ \top , B}\).
    Notice \(\Img(I)\), the range of \(I\), 
    is a relational KAT with the largest element \(I(( \sum  K)^*)\),
    i.e. \(\Img(I)\) is a general relational TopKAT.
    Because \(I\) is a KAT interpretation, 
    it preserves all the KAT operations and the largest element.
    Hence, \(I\) is a TopKAT homomorphism from \(\KAT_{K_ \top , B}\) to \(\Img(I)\).

    Then we can construct \(I  \circ  r: \TopKAT_{K, B}  \to  \Img(I)\),
    a general relational interpretation:
    \begin{align*}
        I(r(t_{1})) = I(r(t_{2}))
         &  \iff  I  \circ  r(t_{1}) = I  \circ  r(t_{2})                           \\
         &  \iff  P(I  \circ  r(p_{1}),  \ldots  , I  \circ  r(p_{n}))
            & \text{\(I  \circ  r\) is a \(\TopGREL\) interpretation} \\
         &  \iff  P(I(p_{1}),  \ldots  , I(p_{n}))
            & r(p_{i}) = p_{i}
    \end{align*}
    Thus the two KAT terms \(r(t_{1}), r(t_{2})  \in  \KAT_{K_ \top , B}\) also can express the predicate \(P\).
\end{proof}

Since the image of \(I\) is not necessarily a relational TopKAT,
where the top element is interpreted as the complete relation,
the above trick does not work for relational TopKAT.
It is also known that relational TopKAT is strictly more expressive than general relational TopKAT,
since relational TopKAT can encode incorrectness logic,
where general relational TopKAT cannot~\cite{Zhang_de_Amorim_Gaboardi_2022}.

\section{(Co)domain Completeness}\label{sec: domain completeness of TopKAT}

In general, TopKAT is not complete over relational models, which are crucial for
applications in program logics~\cite{Zhang_de_Amorim_Gaboardi_2022}.  However, it
was later showed that we can obtain a complete theory for relational
models by simply adding the axiom \(p  \top  p  \geq  p\) to the theory of
TopKAT~\cite{Pous_Wagemaker_2023}. 

In this paper, we take a different approach than Pous et al.~\cite{Pous_Wagemaker_2023}:
instead of extending the TopKAT framework, we will restrict the completeness result.
In particular, the encoding of incorrectness logic and Hoare Logic in TopKAT~\cite{Zhang_de_Amorim_Gaboardi_2022}
relies only on the ability of TopKAT to compare the domain and codomain of two
relations.  This raises the question of whether TopKAT suffices for proving such
properties; that is, whether the following completeness results hold: for
\(t_{1}, t_{2}  \in  \KAT_{K, B}\) (i.e. \( \top \) does not appear in \(t_{1}\) and \(t_{2}\))
\begin{align*}
    \REL  \models  \cod(t_{1})  \geq  \cod(t_{2}) &  \iff  \TopKAT  \models   \top  t_{1}  \geq   \top  t_{2} & \text{codomain completeness} \\
    \REL  \models  \dom(t_{1})  \geq  \dom(t_{2}) &  \iff  \TopKAT  \models  t_{1}  \top   \geq  t_{2}  \top  & \text{domain complete}
\end{align*}

In this section, we prove that these equivalences hold, even without the additional axiom.
However, they do \emph{not} hold if we allow terms that contain top.
For example, let \(t_{1}  \triangleq  p  \top  p\), and \(t_{2}  \triangleq  p\). Since \(p  \top  p  \geq  p\) holds in
relational TopKAT, thus \(\dom(p  \top  p)  \geq  \dom(p)\). 
However, \(p  \top  p  \top   \geq  p  \top \) is not provable in TopKAT, 
because the inequality is not valid with the language interpretation.
The incompleteness of codomain comparison can also be shown using the same example.

\subsection{Codomain completeness}

The core insight to prove the domain completeness result is 
to construct a specific relational interpretation \(h  \circ  i  \circ  G\),
where its codomain is equivalent to the complete TopKAT interpretation \(G  \circ  r\):
\[\cod(h  \circ  i  \circ  G(t)) = G  \circ  r( \top  t),\]
where \(i\) is the natural inclusion homomorphism \(i: \mathcal{G}_{K, B}  \hookrightarrow  \mathcal{G}_{K_ \top , B}\), 
that maps every language to itself;
and \(h\) is the classical embedding of language KAT into relational KAT~\cite{Kozen_Smith_1997},
which we will recall as follows:
\[h(L) = \{(s, s  \diamond  s')  \mid  s  \in  GS, s'  \in  L\}.\]
Although \(i\) will not change the outcome of \(G\),
it will add a new primitive action \( \top \) to the alphabet, hence changing the outcome of \(h\).
Such addition will equate the codomain of \(h  \circ  i  \circ  G(t)\) 
with the complete TopKAT interpretation \(G  \circ  r\) of \( \top  t\).
The proof of this equality is by simply computing both sides of the equation.

\begin{lemma}\label{the: codomain completeness core lemma}
    For any term \(t  \in  \KAT_{K, B}\),
    \[\cod(h  \circ  i  \circ  G(t)) = G  \circ  r( \top  t).\]
\end{lemma}

\begin{proof}
    We explicitly write out the domain and codomain of the functions in
    the relational KAT interpretation \(h  \circ  i  \circ  G\) for the ease of the reader:
    \[\KAT_{K, B}
        \xrightarrow{G} \mathcal{G}_{K, B}
        \xrightarrow{i} \mathcal{G}_{K_ \top , B}
        \xrightarrow{h} \mathcal{P}(\mathcal{G}_{K_ \top , B}  \times  \mathcal{G}_{K_ \top , B}).\]
    In this case, \(h\) is a KAT homomorphism from \(\mathcal{G}_{K_ \top , B}\):
    \[h(S) = \{(s, s  \diamond  s_{1})  \mid  s  \in  GS_{K_ \top , B}, s_{1}  \in  S\}.\]
    Since the reduction \(r\) preserves terms without \( \top \),
    let \(t  \in  \KAT_{K, B}\) (i.e. \(t\) does not contain \( \top \)),
    \[G  \circ  r( \top ) = GS_{K_ \top , B} \\ G  \circ  r(t) = G(t).\]
    Therefore, for any term \(t  \in  \KAT_{K, B}\)
    \begin{align*}
        \cod(h  \circ  i  \circ  G(t))
         & = \{s  \alpha  s_{1}  \mid  s  \alpha   \in  GS_{K_ \top , B},  \alpha  s_{1}  \in  G(t)\} \\
         & = GS_{K_ \top , B}  \diamond  G(t)                          \\
         & = (G  \circ  r( \top ))  \diamond  (G  \circ  r(t))                     \\
         & = G  \circ  r( \top  t). \qedhere
    \end{align*}
\end{proof}

\Cref{the: codomain completeness core lemma} established a connection between 
the codomain operator and the language interpretation of TopKAT.
Then by completeness of the language interpretation, 
we will obtain the completeness of codomain comparison.

\begin{theorem}[Codomain completeness]\label{the: codomain completeness}
    Given two terms \(t_{1}, t_{2}  \in  \KAT_{K, B}\) (i.e. terms without \( \top \)),
    then codomain comparison is complete:
    \begin{align*}
        \REL  \models  \cod(t_{1})  \geq  \cod(t_{2}) &  \iff  \TopKAT  \models   \top  t_{1}  \geq   \top  t_{2}.
    \end{align*}
\end{theorem}

\begin{proof}
    Given the natural inclusion homomorphism: \(i: \KAT_{K, B}  \to  \KAT_{K_ \top , B}\),
    we show that the following are equivalent:
    \begin{enumerate}
        \item \(\REL  \models  \cod(t_{1})  \geq  \cod(t_{2}).\)
        \item \(\cod(h  \circ  i  \circ  G(t_{1}))  \geq  \cod(h  \circ  i  \circ  G(t_{2})).\)
        \item \(\TopKAT  \models   \top  t_{1}  \geq   \top  t_{2}.\)
    \end{enumerate}

    We first show that \(1  \implies  2\), by definition, \(\REL  \models  \cod(t_{1})  \geq  \cod(t_{2})\)
    implies \(\cod(I(t_{1}))  \geq  \cod(I(t_{2}))\) for all relational KAT interpretations \(I\).
    Because \(h  \circ  i  \circ  G\) is a relational KAT interpretation, so \(1  \implies  2\).

    We show \(2  \implies  3\), which uses the equality discussed above, 
    and proved in~\Cref{the: codomain completeness core lemma}:
    \begin{align*}
             & \cod(h  \circ  i  \circ  G(t_{1}))  \geq  \cod(h  \circ  i  \circ  G(t_{2}))           \\
         \iff  {} & G  \circ  r( \top  t_{1})  \geq  G  \circ  r( \top  t_{2})
             & \text{\Cref{the: codomain completeness core lemma}} \\
         \iff  {} & \TopKAT  \models   \top  t_{1}  \geq   \top  t_{2}.
             & \text{Completeness of \(G  \circ  r\)}
    \end{align*}

    Finally, we show \(3  \implies  1\), by \Cref{the: top element represent domain}:
    \[\TopKAT  \models   \top  t_{1}  \geq   \top  t_{2}  \implies  \TopREL  \models   \top  t_{1}  \geq   \top  t_{2}  \implies  \REL  \models  \cod(t_{1})  \geq  \cod(t_{2}). \qedhere\]
\end{proof}

\subsection{Domain completeness}

The domain completeness result can be derived from codomain completeness 
by observing properties of opposite TopKAT and the converse operator \((-)^{ \lor }\), 
both of which we will recall below.

For every TopKAT \(\mathcal{K}\), we can construct the opposite TopKAT \(\mathcal{K}^{\op}\) 
by reversing the multiplication operation, keeping the sorts and other operations unchanged:
\[p \mathbin{\hat{ \cdot }} q  \triangleq  q  \cdot  p,\]
where \(\hat{ \cdot }\) is multiplication in \(\mathcal{K}^{\op}\) and \( \cdot \) is multiplication in \(\mathcal{K}\).
By definition, \((-)^{\op}\) is a involution, that is \({(\mathcal{K}^{\op})}^{\op} = \mathcal{K}\).
Furthermore, \((-)^{\op}\) is a TopKAT functor,
this means all TopKAT homomorphisms \(h: \mathcal{K}  \to  \mathcal{K}'\) 
can be lifted to a TopKAT homomorphism on the opposite TopKAT \(h^{\op}: \mathcal{K}^{\op}  \to  {\mathcal{K}'}^{\op}\). 
The lifting \(h^{\op}\) is point-wise equal to \(h\):
\[ \forall  p  \in  \mathcal{K}, h^{\op}(p)  \triangleq  h(p).\]
The fact that \(h^{\op}\) is a TopKAT homomorphism can be proven by unfolding the definition,
and the functor laws are satisfied because \(h^{\op}\) is point-wise equal to \(h\).

There are two important homomorphisms involving opposite TopKAT:
\begin{align*}
    (-)^{ \lor } & : (X  \times  X)^{\op}  \to  (X  \times  X) &
    \op & : \TopKAT_{K, B}  \to  \TopKAT^{\op}_{K, B} \\  
    (R)^{ \lor } & = \{(b, a)  \mid  (a, b)  \in  R\}, & 
     \forall  p  \in  K + B, \op & (p) = p.
\end{align*}
The \((-)^{ \lor }\) is the relational converse operator, 
the rules of homomorphism can simply be proven by unfolding of definitions.
The crucial property of \((-)^{ \lor }\) is that it flips the domain and codomain:
\begin{equation}\label{the: converse flips domain to codomain}
    \dom(R^{ \lor }) = \cod(R).
\end{equation}
Hence, allowing us to flip the result about codomains and apply it to domains.

\(\op\) is a homomorphism from free TopKAT to its opposite TopKAT;
it can be defined by lifting the embedding function \(K + B  \hookrightarrow  \TopKAT_{K, B}\) on primitives.
Intuitively, given a term \(t  \in  \TopKAT\), 
\(\op(t)\) will flip all the multiplications in \(t\) recursively.
\begin{lemma}\label{the: injectivity of op}
    the left inverse of \(op\) can be obtained by lifting itself through the \((-)^{\op}\) functor,
    \[\op^{\op}: \TopKAT^{\op}  \to  (\TopKAT^{\op})^{\op} = \TopKAT.\]
    Recall \(\op^{\op}\) is pointwise equal to \(\op\), 
    thus \(\op^{\op}  \circ  \op: \TopKAT  \to  \TopKAT\) is the identity interpretation 
    because it preserves all the primitives.
    Thus, \(\op\) has a left inverse, hence it is injective:
    \[t_{1} = t_{2}  \iff  \op(t_{1}) = \op(t_{2}).\]
\end{lemma}

Finally, since the elements in \(\TopKAT^{\op}\) are the same as \(\TopKAT\), 
which are TopKAT terms modulo provable TopKAT equalities,
theorems about TopKAT terms are also true for elements in \(\TopKAT^{\op}\).
In particular, codomain completeness (\Cref{the: codomain completeness})
also holds in \(\TopKAT^{\op}\): 
for all terms \(t_{1}, t_{2}  \in  \TopKAT\),
\begin{equation}\label[equiv]{the: op codomain completeness}
     \top   \cdot  \op(t_{1})  \geq   \top   \cdot  \op(t_{2})  \iff  \REL  \models  \cod(\op(t_{1})) = \cod(\op(t_{2})).
\end{equation}

\begin{theorem}[Domain Completeness]\label{the: domain completeness}
    For all terms \(t_{1}, t_{2}  \in  \KAT\), the following equivalence hold:
    \[\REL  \models  \dom(t_{1}) = \dom(t_{2})  \iff  \TopKAT  \models  t_{1}  \top   \geq  t_{2}  \top .\]
\end{theorem}

\begin{proof}
    \( \impliedby \) direction is trivial by \Cref{the: top element represent domain};  
    and \( \implies \) direction can be derived as follows:
    let \(I\) be some relational interpretation,
    then \(I^{\op}(\op(-))^ \lor \) is also a relational interpretation:
    \[I^{\op}(\op(-))^ \lor : 
        \TopKAT \xrightarrow{\op} \TopKAT^{\op} \xrightarrow{I^{\op}} 
        (X  \times  X)^{\op} \xrightarrow{(-)^{ \lor }} (X  \times  X).\]
    Thus, we let \(I\) range over all relational interpretations:
    \begin{align*}
        & \REL  \models  \dom(t_{1})  \supseteq  \dom(t_{2})  \\
        &  \implies   \forall  I, \dom(I(t_{1}))  \supseteq  \dom(I(t_{2})) \\
        &  \implies   \forall  I, \dom(I^{\op}(\op(t_{1}))^ \lor )  \supseteq  \dom(I^{\op}(\op(t_{2}))^ \lor ) 
            &\text{specialize \(I\) as \(I^{\op}(\op(-))^ \lor \)}\\  
        &  \implies   \forall  I, \cod(I^{\op}(\op(t_{1})))  \supseteq  \cod(I^{\op}(\op(t_{1}))) 
            &\text{\Cref{the: converse flips domain to codomain}}\\
        &  \implies   \forall  I, \cod(I(\op(t_{1})))  \supseteq  \cod(I(\op(t_{1}))) 
            &\text{\(I^{\op}\) is pointwise equal to \(I\)}\\
        &  \implies   \top   \cdot  \op(t_{1})  \geq   \top   \cdot  \op(t_{2}) 
            &\text{\Cref{the: op codomain completeness}}\\
        &  \implies  \op( \top   \cdot  t_{1})  \geq  \op( \top   \cdot  t_{2}) 
            & \text{Definition of \(\op\)}\\
        &  \implies  t_{1}  \top   \geq  t_{2}  \top  & \text{\Cref{the: injectivity of op}}
    \end{align*}
\end{proof}

\begin{remark}
    Alternatively, \Cref{the: domain completeness} can also be proven 
    by constructing the following \(h'\):
    \begin{align*}
        h' & : \mathcal{G}_{K, B}  \to  \mathcal{P}(\mathcal{G}_{K, B}  \times  \mathcal{G}_{K, B})\\
        h' & (S_{1})  \triangleq  \{(s_{1}  \alpha  s,  \alpha  s)  \mid  s_{1}  \alpha   \in  S_{1},  \alpha  s  \in  GS_{K, B}\}.
    \end{align*}
    Then the proof would mirror that of \Cref{the: codomain completeness},
    replacing \(h\) with \(h'\) and replacing \(\cod\) with \(\dom\).
    However, the proof of \Cref{the: domain completeness} reveals more properties
    of maps like \((-)^{ \lor }\) and \(\op\), 
    thus we choose to present the current proof of \Cref{the: domain completeness} 
    instead of the alternative proof.
\end{remark}

\section{Related Works}


\itemTitle{Extensions of Kleene algebra and reduction:}
soon after the completeness of Kleene algebra was proven~\cite{Kozen_1994},
it was realized that adding an embedded Boolean algebra can help reasoning
about control structures, such system is referred to as
Kleene algebra with tests (KAT)~\cite{Kozen_Smith_1997,Cohen_Kozen_Smith_1999}.
Later KAT was further extended to reason about failure~\cite{Mamouras_2017},
indicator variables~\cite{Grathwohl_Kozen_Mamouras_2014},
domain~\cite{Desharnais_Möller_Struth_2006}, networks~\cite{Anderson_Foster_Guha_Jeannin_Kozen_Schlesinger_Walker_2014},
and relational reasoning~\cite{Antonopoulos_Koskinen_Le_Nagasamudram_Naumann_Ngo_2022}.
Kleene algebra has also been extended to reason about 
concurrency, as concurrent Kleene algebra~\cite{Hoare_van_Staden_Möller_Struth_Zhu_2016, Kappé_Brunet_Silva_Zanasi_2018}
and concurrent Kleene algebra with observations~\cite{Kappé_Brunet_Silva_Wagemaker_Zanasi_2020}.
Many of these extensions can be seen as Kleene algebra with extra hypotheses~\cite{Cohen_1995,Doumane_Kuperberg_Pous_Pradic_2019}.
Although many hypotheses make the theory undecidable~\cite{Kozen_1996,Kozen_2002,Doumane_Kuperberg_Pous_Pradic_2019},
many useful hypotheses can be eliminated via reduction~\cite{Pous_Rot_Wagemaker_2021}.
Thus, our new perspective on reduction could potentially lead to streamlining of various previous proofs, 
and more general proofs of completeness results.

\itemTitle{Top element:}
Tarski's relational algebra~\cite{tarski_CalculusRelations_1941} contains the addition, 
mulitiplication, and identity operation of KA;  
in addition, relational algebra also include a top element. 
Hence attempts to incorporat Kleene star into relational algebra 
effectively create a super theory of TopKAT.
Unfortuantly, several attempts at these algebras turn out to be undecidable
because of the presence of intersection and 
converse operations~\cite{andrekaAxiomatizabilityPositiveAlgebras2011, pous_PositiveCalculusRelations_2018}.
With the intersection and converse operators removed, 
top element is proven to be individually useful in Kleene algebra:
for example, Mamouras~\cite{Mamouras_2017} uses the top element to forget program states,
and Antonopoulos et al.~\cite{Antonopoulos_Koskinen_Le_Nagasamudram_Naumann_Ngo_2022} 
uses top to design forward simulation rules for relational verification, 
and claim that relational incorrectness logic~\cite{murray_UnderApproximateRelationalLogic_2020a} 
can be encoded using BiKAT extended with top.
The completeness and decidability of TopKAT was first studied by Zhang et al.~\cite{Zhang_de_Amorim_Gaboardi_2022},
and concluded that TopKAT is not complete with relational models.
Later, Pous et al.~\cite{Pous_Wagemaker_2022,Pous_Wagemaker_2023} showed that 
both TopKA and TopKAT is complete with relational model with one additional axiom: \(p  \top  p  \geq  p\),
and the theory remains PSPACE-complete, like KAT and TopKAT.
In this paper, we showed that TopKAT without the additional axiom is complete 
for a specific form of inequalities, namely when top only appears in the front or the end of the term.
Although this form of inequalities seem restrictive, 
they are enough to encode both Hoare and incorrectness logic~\cite{Zhang_de_Amorim_Gaboardi_2022}.

\itemTitle{Domain in KAT:}
The study of axiomatizing (co)domain in KAT has a long and rich history. 
Domain semiring~\cite{Desharnais_Struth_2011} 
and Kleene algebra with domain~\cite{Desharnais_Möller_Struth_2006}
were two popular yet different axiomatizations of (co)domain in Kleene algebra with tests.
These two axiomitizations turn out to coincide in a large class of semirings~\cite{Fahrenberg_Johansen_Struth_Ziemiánski_2021}.
Various applications for domain in KAT have been discovered, including modeling
program correctness, predicate transformers, temporal logics, 
termination analysis, and many more~\cite{Desharnais_Möller_Struth_2004}.
Many of these applications can even be efficiently automated~\cite{hofner_AutomatedReasoningKleene_2007}.
However, although the free relational model of these theories has been characterized~\cite{mclean_FreeKleeneAlgebras_2020},
the search for general complete interpretation remains unfruitful.
The complexity of these theories was recently shown to be EXPTIME-complete~\cite{Sedlár_2023},
a worse complexity class than PSPACE-complete for TopKAT.

\section{Conclusion And Open Problems}

In this paper, we exploit the homomorphic structure of reduction
to simplify the proof of various previous results~\cite{Zhang_de_Amorim_Gaboardi_2022}.
We have also showed that TopKAT is complete with respect to (co)domain comparison
in the relational models,
which lays a solid foundation for the use of TopKAT in (co)domain reasoning.

However, there are still several interesting unsolved problems about TopKAT.
Most of the incorrectness logic rules are written using hypotheses,
for example, the sequencing rule:
\[
    \frac{[a]~p~[b] \qquad [b]~q~[c]}{[a]~p  \cdot  q~[c]}
\]
corresponds to the implication \( \top  a p  \leq   \top  b  \land   \top  b p  \leq   \top  c  \implies   \top  a p q  \leq   \top  c\).
Although each individual inequality in the implication fits the desired form \( \top  t_{1}  \geq   \top  t_{2}\).
it is unclear whether implications of the form
\[ \top  t_{11}  \leq   \top  t_{12}  \land   \top  t_{21}  \leq   \top  t_{22}  \land   \cdots   \land   \top  t_{n1}  \leq   \top  t_{n2}  \implies   \top  t_{1}  \leq   \top  t_{2}\]
are complete with relational TopKAT or decidable.

Recently, there is an efficient fragment of KAT proposed, named 
\emph{Guarded Kleene algebra with tests}~\cite{Smolka_Foster_Hsu_Kappé_Kozen_Silva_2020}
or \emph{GKAT}.
This fragment not only enjoys nearly-linear time equality checking,
but also soundly models probabilistic computations as well. 
It would be interesting to see whether the completeness and decidability result of TopKAT
can be extended to GKAT, and whether the efficiency of GKAT will persist with the addition of top.

\bibliography{ref}

\end{document}